\documentclass{amsart}
\usepackage{latexsym, amscd, amsfonts, mathrsfs, amsmath, amssymb, amsthm, stmaryrd, tikz-cd, mathrsfs, bbm, url}

\def\p{\prime}
\def\pp{{\prime\prime}}
\def\bZ{\mathbb Z}
\def\bC{\mathbb C}
\def\bR{\mathbb R}

\def\bF{\mathbb F}

\def\bbl{\mathbbm 1}
\def\isom{\simeq}

\def\sse{\subseteq}

\def\ol{\overline}

\def\lm{\lambda}

\def\cC{\mathcal C}
\def\cS{\mathcal S}
\def\cT{\mathcal T}

\def\al{\alpha}
\def\be{\beta}

\def\t{\times}
\def\mto{\mapsto}

\def\Th{\Theta}

\def\bup{\bigcup}

\def\bP{\mathbb P}

\def\ep{\epsilon}

\def\der{\partial}

\def\cA{\mathcal{A}}

\DeclareMathOperator\ch{\mathrm{char}}

\DeclareMathOperator\Conf{\mathrm{Conf}}
\DeclareMathOperator\SConf{\mathrm{SConf}}

\newcommand{\rom}[1]{\textup{\uppercase\expandafter{\romannumeral#1}}}

\newtheorem{thm}{Theorem}[section]
\newtheorem{lemma}[thm]{Lemma}
\newtheorem{prop}[thm]{Proposition}
\newtheorem{coro}[thm]{Corollary}

\theoremstyle{definition}
\newtheorem{defn}[thm]{Definition}
\newtheorem{eg}[thm]{Example}
\newtheorem{conj}[thm]{Conjecture}

\newtheorem{rmk}[thm]{Remark}

\begin{document}
  \title{Zero-error communication over adder MAC}
  \author{Yuzhou Gu}
  \begin{abstract}
Adder MAC is a simple noiseless multiple-access channel (MAC), where if users send messages $X_1,\ldots,X_h\in \{0,1\}^n$, then the receiver receives $Y = X_1+\cdots+X_h$ with addition over $\bZ$.
Communication over the noiseless adder MAC has been studied for more than fifty years. There are two models of particular interest: uniquely decodable code tuples, and $B_h$-codes.
In spite of the similarities between these two models, lower bounds and upper bounds of the optimal sum rate of uniquely decodable code tuple asymptotically match as number of users goes to infinity, while there is a gap of factor two between lower bounds and upper bounds of the optimal rate of $B_h$-codes.

The best currently known $B_h$-codes for $h\ge 3$ are constructed using random coding. In this work, we study variants of the random coding method and related problems, in hope of achieving $B_h$-codes with better rate. Our contribution include the following.
  \begin{enumerate}
  \item We prove that changing the underlying distribution used in random coding cannot improve the rate.
  \item We determine the rate of a list-decoding version of $B_h$-codes achieved by the random coding method.
  \item We study several related problems about R\'{e}nyi entropy.
  \end{enumerate}
  \end{abstract}
  \maketitle
  \tableofcontents
  \section{Introduction}
  \subsection{Overview}
  Adder MAC is a simple noiseless multiple-access channel (MAC), where if users send messages $X_1,\ldots,X_h\in \{0,1\}^n$, then the receiver receives $Y = X_1 + \cdots + X_h$ with addition over $\bZ$.
  
  Communication over the noiseless adder MAC has been studied for more than fifty years.
  In the most well-studied version, each user $i$ ($1\le i\le h$) has their own codebook $\cC_i\sse \{0,1\}^n$, and $X_i$ is picked from $\cC_i$.
  We insist our protocol to be zero-error, i.e., we can uniquely determine $X_1,\ldots,X_k$ given $Y$. More formally, if we have $u_1,\ldots,u_h, v_1,\ldots,v_h$ with $u_i,v_i\in \cC_i$ ($1\le i\le h$) and $$u_1+\cdots+u_h = v_1+ \cdots + v_h,$$
  then we must have $u_i=v_i$ for all $1\le i\le h$.
  A code tuple $(\cC_1,\ldots,\cC_h)$ satisfying the above property is called uniquely decodable.
  The quantity we would like to optimize is the sum rate, defined as
  \begin{align*}
  R(\cC_1,\ldots,\cC_h) = \sum_{1\le i\le h} \frac{\log |\cC_i|}n
  \end{align*}
  where the logarithm is taken over base $2$.
  Let $$R_h = \limsup_{n\to \infty} \sup_{\substack{\cC_1,\ldots,\cC_h\sse \{0,1\}^n\\ \text{uniquely decodable}}} R(\cC_1,\ldots,\cC_h).$$
  
  By standard information theory (e.g., \cite{PW17} Chapter 29), we have 
  \begin{align*}
  R_h \le H(B(h,\frac 12)) = (1+o(1)) \frac{\log h}2
  \end{align*}
  where $B(h,\frac 12)$ is the binomial distribution and $H$ is Shannon entropy.
  On the other hand, Cantor and Mills \cite{CM66} and Lindstr\"{o}m \cite{Lin65} constructed code tuples whose sum rate grow as $(1+o(1)) \frac {\log h}2$ as $h\to \infty$. Therefore lower bound and upper bound match.
  
  There is another version of communication over the noiseless adder MAC, where all the users share a single codebook $\cC\sse \{0,1\}^n$.
  In this case, we cannot expect to be able to uniquely determine $X_1,\ldots,X_h$ given $Y$, because permuting $X_i$'s does not change $Y$.
  Instead, we require that we can uniquely determine the multiset $\{X_1,\ldots,X_h\}$ given $Y$.
  Formally, if we have $u_1,\ldots,u_h, v_1,\ldots,v_h\in \cC$ and 
  $$u_1+\cdots+u_h=v_1+\cdots+v_h,$$
  then the multisets $\{u_1,\ldots,u_h\}$ and $\{v_1,\ldots,v_h\}$ are equal.
  Codes $\cC$ satisfying this property are called $B_h$-codes.
  
  In this setting the quantity we would like to optimize is the rate $$R(\cC) = \frac{\log |\cC|}n.$$ We define 
  \begin{align*}
  R^*_h = \limsup_{n\to \infty} \sup_{\substack{\cC\sse \{0,1\}^n\\ \text{$B_h$-code}}} R(\cC).
  \end{align*}
  Again, standard information theory gives
  \begin{align*}
  R^*_h \le \frac 1h H(B(h, \frac 12)) = (1+o(1)) \frac{\log h}{2h}.
  \end{align*}
  The best known lower bound so far is $R_2^*\ge \frac 12$ given by Lindstr\"{o}m \cite{Lin69} and 
  $$R_h^*\ge \frac{\log \left (\frac {2^{2h}}{\binom {2h}h} \right)}{2h-1} = (1+o(1)) \frac{\log h}{4h}$$
  for $h\ge 3$ given by Poltyrev \cite{Pol87}.
  Therefore there is a gap of factor $2$ between the lower bound and the upper bound.
  
  Poltyrev's construction is based on random coding. In this work, we study variants of the random coding method and related problems, in hope of achieving $B_h$-codes with better rate. Our contribution includes the following.
  \begin{enumerate}
  \item We determine the rate achieved by changing the underlying distribution used in random coding.
  \item We determine the rate of a list-decoding version of $B_h$-codes achieved by the random coding method.
  \item We study several related problems about R\'{e}nyi entropy.
  \end{enumerate}
  
  \subsection{Related work}
  In this section we review previous works on uniquely decodable code tuples and $B_h$-codes.
  
  In both settings, the case $h=2$ is studied most. 
  A uniquely decodable code tuple with $h=2$ is called uniquely decodable code pair (UDCP) in literature.
  Lindstr\"{o}m \cite{Lin69} prove that $\frac 12 \log 6 \le R_2 \le \frac 32$. Since then, a lot of constructions of UDCPs have been given, improving the lower bound on $R_2$, including $1.30366$ by Coebergh van den Braak and van Tilborg \cite{CV85}, $1.30369$ by Ahlswede and Balakirsky \cite{AB99}, $1.30565$ by Coebergh van den Braak \cite{Coe83}, $1.30999$ by Urbanke and Li \cite{UL98}, $1.31782$ by Mattas and \"Osterg\r{a}rd \cite{MO05}.
  These constructions are all explicit constructions, and usually have small $n$.
  There has been no upper bound of $R_2$ better than the entropy bound $\frac 32$.
  
  For UDCPs, people have also considered the following question: when $\al = \frac{\log |\cC_1|}n$ is close to one, how large can $\be = \frac{\log |\cC_2|}n$ be?
  Kasami et al. \cite{KLWY83} gave a construction where $\al \ge 1-\ep$ and $\be \ge 0.25$.
  This is recently improved by Wiman \cite{Wim17} to $\al \ge 1-\ep$ and $\be \ge 0.2563$.
  Urbanke and Li \cite{UL98} proved that when $\al \ge 1-\ep$, we have $\be \le 0.4921$.   It is improved by Ordentlich and Shayevitz \cite{OS16} to $\be \le 0.4798$ when $\al \ge 1-\ep$ and by Austrin et al. \cite{AKKN17} to $\be \le 0.4228$ when $\al \ge 1-\ep$.
  
  Let us discuss $R_h$ for $h\ge 3$. The special case where all $|\cC_i|=2$ is studied under the name detecting matrix.
  Cantor and Mills \cite{CM66} and Lindstr\"{o}m \cite{Lin65} constructed codes with sum rate increases as $(1+o(1)) \frac{\log h}2$ as $h\to \infty$.
  For the general case, Khachatrian and Martirossian \cite{KM98} gave a combinatorial construction for all $h$. 
  Kiviluoto and {\"{O}}sterg{\r{a}}rd \cite{KO07} gave better explicit constructions for $3\le h\le 5$.
  For $k\ge 3$, it has been proven by Bross and Blake \cite{BB98} that $R_h$ is strictly smaller than $H(B(h,\frac 12))$.
  
  Now let us turn to $R_2^*$. 
  A $B_2$-code is also called a Sidon code, in analogy with Sidon sequences in number theory.
  Lindstr\"{o}m \cite{Lin69} gave a construction of a $B_2$ code of rate $\frac 12$ that actually works with addition over $\bZ/2\bZ$.
  There has been several nontrivial upper bounds for $R_2^*$.
  Lindstr\"{o}m proved in \cite{Lin69} that $R_2^*\le \frac 23$, and in \cite{Lin72} that $R_2^* \le 0.6$.
  Cohen et al. \cite{CLZ01} improved this to $R_2^* \le 0.5753$.
  
  The number theoretic analogy of $B_h$-codes are called $B_h$-sequences. 
  Any construction of $B_k$-sequences can be directly translated into $B_k$-codes with the same rate.
  Bose and Chowla \cite{BC62}, using finite fields, constructed $B_h$-sequences (and thus $B_h$-codes) with rate $\frac 1h$.
  This rate is optimal in number theoretic setting, but known to be suboptimal in coding theoretic setting, at least for $h\ge 3$.
  D'yachkov and Rykov \cite{DR81} proved using random coding that $R_h^* \ge \frac{\log\left(\frac{2^{2h}}{\binom{2h}{h}}\right)}{2h}$.
  This is improved by Poltyrev \cite{Pol87} to $R_h^* \ge \frac{\log\left(\frac{2^{2h}}{\binom{2h}{h}}\right)}{2h-1}$.
  As $h$ goes to $\infty$, the above rates increase as $(1+o(1)) \frac{\log h}{4h}$.
  There has been no general upper bound for $R_h^*$ with $h\ge 3$ except for the trivial fact that $R_h^*\le R_h$. 
  
  D'yachkov-Rykov \cite{DR81} also studied what they called plans, a weaker version of $B_h$-codes, which are sets $\cC\sse \{0,1\}^n$ satisfying the property that 
  for two distinct subsets $\{u_1,\ldots,u_h\}$ and $\{v_1,\ldots,v_h\}$ of $\cC$, we have
  \begin{align*}
  u_1 + \cdots + u_h = v_1 + \cdots + v_h.
  \end{align*}
  The currently known lower bound and upper bound of the optimal rate of plans are the same as those of $B_h$-codes.
  \subsection{Acknowledgements}
  I am deeply grateful to my advisor Yury Polyanskiy for his support and guidance in various aspects. Without him I would not have learnt the beauty of information theory.
I would like to thank Ganesh Ajjanagadde, Alexey Frolov, and all people who has discussed the material with me for their helpful advice.
I also thank my family and friends, in particular my partner Zhulin Li, for their support and encouragement.
  \section{Preliminaries}
  \subsection{Definitions}
  In this chapter we give necessary definitions and study existing constructions of $B_h$-codes.
  \begin{defn}\label{DefnBhSet}
  Let $A$ be an abelian group and $\cC\sse A$ be a subset.
  We say $\cC$ is a $B_h$-set (or $\cC$ satisfies the $B_h$-property) if for any $a\in A$, there exists at most one multiset $\{u_1,\ldots,u_h\}$ with $u_1,\ldots,u_h\in \cC$ such that $a = u_1 + \cdots + u_h$.
  In other words, if $u_1,\ldots,u_h,v_1,\ldots,v_h\in \cC$ satisfies 
  $u_1+\cdots+u_h=v_1+\cdots+v_h$, then 
  the multisets $\{u_1,\ldots,u_h\}$ and $\{v_1,\ldots,v_h\}$ are equal.
  \end{defn}
  \begin{defn}\label{DefnBhCode}
  A $B_h$-code is a $B_h$-set $\cC \sse \{0,1\}^n\sse \bZ^n$.
  The rate of a $B_h$-code is defined as $\frac{\log |\cC|}n$.
  \end{defn}
  \begin{rmk}
  We are most interested in the asymptotic growth of rate of $B_h$-codes as $n$ goes to $\infty$.
  Therefore, when we say ``there exist $B_h$-codes of rate $f$'' we actually mean ``there exist a family of $B_h$-codes $\cC_1, \cC_2,\cdots$ with $\cC_i\sse \{0,1\}^{n_i}$ such that $n_i\to \infty$ and $\log \frac{\log |\cC_i|}{n_i}\to f$ as $i\to \infty$.''
  \end{rmk}
  \subsection{Constructions from number theory}
  The only known explicit constructions of $B_h$-codes achieve rate $\frac 1h$, and all come form number theory.
  \begin{thm}[Bose-Chowla \cite{BC62}]
  Let $q$ be a prime power and $h$ be a positive integer.
  Then there exists a $B_h$-set $\cC\sse \bZ/(q^h-1)\bZ$ of size $q$.
  \end{thm}
  \begin{proof}
  Let $\al \in \bF_{q^h}$ be a generator of $\bF_{q^h}$ over $\bF_q$.
  By elementary number theory, $\deg \al = h$ and $\al$ is a generator of $\bF_{q^h}^\t \isom \bZ/(q^h-1)\bZ$.
  Let elements of $\bF_q$ be $x_1,\ldots,x_q$.
  Let $d_i \in \bZ/(q^h-1)\bZ$ be the unique solution to $\al^{d_i} = \al + x_i$.
  We claim that $\cC = \{d_1,\ldots,d_q\}$ is a $B_h$-set.
  
  Suppose that we have $u_1,\ldots,u_h,v_1,\ldots,u_h$ with $1\le u_i,v_i\le q$ ($1\le i\le h$) satisfying $$d_{u_1} + \cdots + d_{u_h} = d_{v_1} + \cdots + d_{v_h}.$$
  Then we have 
  $$\al^{d_{u_1} + \cdots + d_{u_h}} = \al^{d_{v_1} + \cdots + d_{v_h}}.$$
  By definition of $d_i$, this means 
  $$(\al+x_{u_1})\cdots (\al+x_{u_h}) = (\al+x_{v_1})\cdots(\al+x_{v_h}).$$
  Consider the polynomial $$f(x) = (x+x_{u_1})\cdots (x+x_{u_h}) - (x+x_{v_1})\cdots(x+x_{v_h}).$$
  Because $\deg f \le h-1$ and $f(\al)=0$, we must have $f=0$.
  So the multisets $\{u_1,\ldots,u_h\}$ and $\{v_1,\ldots,v_h\}$ are equal.
  \end{proof}
  \begin{coro}
  There exist $B_h$-codes of rate $\frac 1h$.
  \end{coro}
  \begin{proof}
  Let elements of $\bZ/(q^h-1)\bZ$ be $0,1,\ldots,q^h-2$.
  Consider the map $f: \bZ/(q^h-1)\bZ \to \{0,1\}^{\lceil \log_2 (q^h-2)\rceil}$ which maps an integer to its binary representation.
  It is easy to see that $f$ preserves the $B_h$-property: for any $B_h$-set $\cC\sse \bZ/(q^h-1)\bZ$, its image $f(\cC)$ is also a $B_h$-set.
  Therefore we get a $B_h$-set of rate $\frac{\log q}{\lceil \log (q^h-2)\rceil} = (1+o(1)) \frac 1h$.
  As $q\to \infty$ we get the desired code family.
  \end{proof}
  We present another (folklore) construction of $B_h$-codes which has nice geometric meaning. This construction is similar to the construction in Lindstr\"{o}m \cite{Lin69} for $h=2$.
  \begin{thm}
  Let $\bF_q$ be a finite field with $\ch \bF_q > h$. Then there exists a $B_h$-set $\cC\sse \bF_q^h$ of size $q$.
  \end{thm}
  \begin{proof}
  Let $\cC = \{(x^1,x^2,\cdots,x^h):x\in \bF_q\}$. We claim that $\cC$ is a $B_h$-set.
  Suppose we have $u_1,\ldots,u_h,v_1,\ldots,v_h\in \bF_q$ such that for $i=1,2,\ldots,h$ we have
  \begin{align*}
  u_1^i + \cdots + u_h^i = v_1^i + \cdots + v_h^i.
  \end{align*}
  By Newton's identities for symmetric polynomials, we see that $$e_i(u_1,\ldots,u_h) = e_i(v_1,\ldots,v_h)$$ for $0\le i\le h$, where $e_i$ is the $i$-th elementary symmetric polynomial.
  So we have
  \begin{align*}
  (x+u_1)\cdots (x+u_h) &= \sum_{0\le i\le n} e_i(u_1,\ldots,u_h) x^{n-i} \\
  & = \sum_{0\le i\le n} e_i(v_1,\ldots,v_h) x^{n-i} \\
  & = (x+v_1)\cdots (x+v_h).
  \end{align*}
  Therefore the multisets $\{u_1,\ldots,u_h\}$ and $\{v_1,\ldots,v_h\}$ are equal.
  \end{proof}
  This construction also implies that there exist $B_h$-codes of rate $\frac 1h$.
  \subsection{Random coding for $B_h$-code}
  For $h\ge 3$, the best currently known $B_h$-codes are all inexplicit and constructed by random coding. We review the construction and formulate the proof in a way so that it can be easily generalized to more complicated constructions.
  \subsubsection{D'yachkov-Rykov}
  \begin{thm}[D'yachkov-Rykov \cite{DR81}] \label{ThmRateBh}
  There exist $B_h$-codes of rate $\frac{\log\left(\frac{2^{2h}}{\binom{2h}h}\right)}{2h}$.
  \end{thm}
  Fix vector length $n$ and number of vectors $t$.
  Let $v_1,\ldots,v_t\in \{0,1\}^n$ be iid uniformly randomly chosen.
  Let $\cC = \{v_1,\ldots,v_t\}$.
  
  Let us consider the probability that $\cC$ is a $B_h$-code.
  Suppose $\cC$ is not a $B_h$-code. Then there exist $i_1,\ldots,i_h,j_1,\ldots,j_h$ such that 
  $$v_{i_1} + \cdots + v_{i_h} = v_{j_1} + \cdots v_{j_h}$$
  and the multisets $\{i_1,\ldots,i_h\}$ and $\{j_1,\ldots,j_h\}$ are not equal.
  
  One immediate idea is to bound the expected number of such violations of $B_h$-property. If the expectation is smaller than one, then we know that there exist desired $B_h$-codes.
  However, this idea does not work, because the expectation could be large.
  For example, if $v_1=v_2$, then we have $\Th(t^{h-1})$ violations of the form $$v_1+v_{i_2} + \cdots+v_{i_h} = v_2 + v_{i_2} + \cdots + v_{i_h}.$$
  Therefore, instead of looking at the expected number of violations, we bound the expected number of ``minimal'' violations, i.e., $i_1,\ldots,i_k, j_1,\ldots,j_k$ ($1\le k\le h$) such that $$v_{i_1} + \cdots + v_{i_k} = v_{j_1} + \cdots v_{j_k}$$
  and the multisets $\{i_1,\ldots,i_k\}$ and $\{j_1,\ldots,j_k\}$ are disjoint.
  
  Furthermore, minimal violations with the same $k$ can have different forms. For example, the probability that $v_1+v_1=v_2+v_3$ is different from the probability that $v_1+v_2=v_3+v_4$. To address this, we make the following definition.
  \begin{defn}\label{DefnConfk2}
  A configuration $C$ of shape $(k,2)$ is a $k\t 2$ matrix of random variables $(C_{i,j})_{1\le i\le k, 1\le j\le 2}$ taking values in $\{0,1\}$ with the property that
  \begin{enumerate}
  \item For each $i,j$, $\bP(C_{i,j}=0) = \bP(C_{i,j}=1)=\frac 12$.
  \item Some (or no) variables are identified, i.e., $\bP(C_{i,j}=C_{i^\p,j^\p})=1$ for some $i,j,i^\p,j^\p$. We treat identified variables as the same variable. Variables that are not identified are mutually independent.
  \item No variable appears in two columns, i.e., if $\bP(C_{i,j}=C_{i^\p,j^\p})=1$, then $j=j^\p$.
  \end{enumerate}
  Two configurations of the same shape are equivalent if they have the same law after repeatedly (1) swapping columns and (2) swapping entries in the same column.
  Let $\Conf(k,2)$ denote the set of equivalence classes of configurations of shape $(k,2)$.
  Let $\Conf(\le h,2) = \bup_{1\le k\le h} \Conf(k,2)$.
  
  Define $d(C)$ to be the number of distinct variables in $C$.
  Define $p(C)$ to be the probability that $$C_{1,1}+\cdots+C_{k,1} = C_{1,2} + \cdots + C_{k,2}.$$
  \end{defn}
  \begin{rmk}
  Due to the equivalence condition, we can also define a configuration of type $(k,2)$ as two disjoint size-$k$ multisets of random variables.
  Similarly, in Definition \ref{DefnConfkl}, we can define a configuration of type $(k,l)$ as a size-$l$ set of size-$k$ multisets of random variables satisfying certain properties.
  We choose to describe a configuration as a matrix because this is easier to present.
  \end{rmk}
  \begin{eg}
  There is one configuration of shape $(1,2)$: $C = \begin{pmatrix}a & b\end{pmatrix}$. (We use different lowercase letters to denote distinct variables.)
  We have $d(C) = 2$ and $p(C) = \frac 12$.
  
  There are three non-equivalent configurations of shape $(2,2)$.
  They are
  \begin{enumerate}
  \item $C_1=\begin{pmatrix}a & b \\ a & b\end{pmatrix}$. $d(C_1) = 2$ and $p(C_1) = \frac 12$.
  \item $C_2=\begin{pmatrix}a & b \\ a & c\end{pmatrix}$. $d(C_2) = 3$ and $p(C_2) = \frac 14$.
  \item $C_3=\begin{pmatrix}a & b \\ c & d\end{pmatrix}$. $d(C_3) = 4$ and $p(C_3) = \frac 38$.
  \end{enumerate}
  
  It is not hard to see that there are $\binom{p_k+1}2$ non-equivalent configurations of shape $(k,2)$, where $p_k$ is the number of partitions of $k$.
  \end{eg}
  Now let us discuss the relationship between minimal violations and configurations.
  For each minimal violation $i_1,\ldots,i_k,j_1,\ldots,j_k$, we can associate it with a configuration of shape $(k,2)$, by identifying $C_{a,1}$ and $C_{b,1}$ for $i_a=i_b$, and identifying $C_{a,2}=C_{b,2}$ for $j_a=j_b$.
  Simple calculation shows that for each configuration $C$, there are $\Th(t^{d(C)})$ minimal violations associated with it, and for each such minimal violation, the probability that it occurs is $p(C)^n$.

  Therefore the expected number of minimal violations is at most
  \begin{align*}
  c\cdot \sum_{C \in \Conf(\le h,2)} t^{d(C)} p(C)^n
  \end{align*}
  where $c$ is a constant only depending on $h$.
  So when $$t = c^\p \left(\max_{C\in \Conf(\le h,2)} p(C)^{1/d(C)}\right)^{-n}$$ for some small enough constant $c^\p$, the expected number of minimal violations is less than one. So the only problem remains is to determine the maximum value of $p(C)^{1/d(C)}$ for $C\in \Conf(\le h,2)$.
  
  It turns out that the maximum value is achieved at the configuration whose all variables are distinct, i.e., $d(C) = 2h$. Let $C_{\max}(h,2)$ denote this configuration.
  
  We need the following lemmas.
  \begin{lemma}\label{LemmaMaxConfh2}
  Let $X = \sum_{1\le i\le d} c_i X_i$ where $c_i\in \bZ_{\ge 1}$, $X_1,\ldots,X_d$ are iid uniform random variables taking values in $\{0,1\}$. 
  Then $$\sum_{a\ge 0} \bP(X=a)^2 \le p(C_{\max}(d,2)).$$
  \end{lemma}
  \begin{proof}
  Let $Y = \sum_{1\le i\le d} c_i Y_i$ where $Y_1,\ldots,Y_d$ are an independent copy of $X_1,\ldots,X_d$.
  Then $\sum_{a\ge 0} \bP(X=a)^2 = \bP(X=Y)$.
  
  Let us consider the characteristic function.
  Because $X-Y$ only takes integer values, we have
  \begin{align*}
  \bP(X=Y) &= \frac 1{2\pi}\int_{-\pi}^{\pi} \phi_{X-Y}(t) dt \\
  & = \frac 1{2\pi}\int_{-\pi}^{\pi} \prod_{1\le i\le d} \phi_{c_i(X_i-Y_i)}(t) dt\\
  & = \frac 1{2\pi}\int_{-\pi}^{\pi} \prod_{1\le i\le d} \phi_{X_i-Y_i}(c_i t) dt.
  \end{align*}
  Note that 
  \begin{align*}
  \phi_{X_i-Y_i}(t) = \phi_{X_i}(t) \ol{\phi_{Y_i}(t)} = |\phi_{X_1}(t)|^2 \in \bR_{\ge 0}.
  \end{align*}
  So we can apply AM-GM and get 
  \begin{align*}
  \bP(X=Y) &\le \frac 1{2\pi d}\sum_{1\le i\le d} \int_{-\pi}^{\pi} \phi_{X_i-Y_i}(c_i t)^d dt \\
  & = \frac 1{2\pi d}\sum_{1\le i\le d} \int_{-\pi}^{\pi} \phi_{X_i-Y_i}(t)^d dt \\
  & = \frac 1{2\pi}\int_{-\pi}^{\pi} \phi_{X_1-Y_1}(t)^d dt \\
  & = \frac 1{2\pi}\int_{-\pi}^{\pi} \phi_{X_1+\cdots+X_d-Y_1-\cdots-Y_d}(t) dt\\
  & = \bP(X_1+\cdots+X_d = Y_1+\cdots+Y_d) \\
  & = p(C_{\max}(d,2)).
  \end{align*}
  \end{proof}
  
  \begin{lemma}\label{LemmaMonoh2}
  The value $p(C_{\max}(d,2))^{1/(2d)}$ is monotone increasing in $d$.
  \end{lemma}
  \begin{proof}
  Let $X_1,\ldots,X_d,Y_1,\ldots,Y_d$ be iid uniform random variables taking values in $\{0,1\}$. 
  Then 
  \begin{align*}
  p(C_{\max}(d,2)) &= \bP(X_1+\cdots+X_d = Y_1 + \cdots+Y_d) \\
  & = \frac 1{2\pi} \int_{-\pi}^{\pi} \phi_{X_1-Y_1}(t)^d dt.
  \end{align*}
  So the lemma follows from generalized mean inequality.
  \end{proof}
  
  Using the lemmas we can prove the following proposition.
  \begin{prop}\label{PropMaxConfh2}
  Over all configurations in $C\in \Conf(\le h,2)$, the configuration $C_{\max}(h,2)$ gives the maximum $p(C)^{1/d(C)}$.
  \end{prop}
  \begin{proof}
  Let $C\in \Conf(k,2)$ where $1\le k\le h$.
  For $a \in \{0,1\ldots,k\}$, let $p_i(a)$ ($i=1,2$) denote the probability that $C_{1,i} + \cdots + C_{k,i} = a$.
  By Cauchy-Schwarz inequality, we have
  \begin{align*}
  p(C) = \sum_{0\le a\le k} p_1(a)p_2(a)
  \le \sqrt{(\sum_{0\le a\le k} p_1(a)^2) (\sum_{0\le a\le k} p_2(a)^2)}.
  \end{align*}
  
  Let $d_i$ ($i=1,2$) denote the number of distinct variables in column $i$.
  By Lemma \ref{LemmaMaxConfh2} and Lemma \ref{LemmaMonoh2},
  \begin{align*}
  \sum_{0\le a\le k} p_i(a)^2 \le p(C_{\max}(d_i,2)) \le p(C_{\max}(h,2))^{d_i/h}.
  \end{align*}
  So we have
  \begin{align*}
  p(C) &\le \sqrt{(\sum_{0\le a\le k} p_1(a)^2) (\sum_{0\le a\le k} p_2(a)^2)}\\
  & \le \sqrt{p(C_{\max}(h,2))^{d_1/h} p(C_{\max}(h,2))^{d_2/h}} \\
  & = p(C_{\max}(h,2))^{d(C)/(2h)}.
  \end{align*}
  In other words, 
  \begin{align*}
  p(C)^{1/d(C)} \le p(C_{\max}(h,2))^{1/(2h)}.
  \end{align*}
  \end{proof}
  
  Now we can prove the theorem.
  \begin{proof}[Proof of Theorem \ref{ThmRateBh}]
  By Proposition \ref{PropMaxConfh2}, we have $$t = c^\p p(C_{\max}(h,2))^{-n/d(C_{\max}(h,2))}$$ and the rate of code $\cC$ is
  \begin{align*}
  \frac{\log t}{n} = (1+o(1)) \frac{-\log p(C_{\max}(h,2))}{d(C_{\max}(h,2))} = (1+o(1))\frac{\log\left(\frac{2^{2h}}{\binom{2h}h}\right)}{2h}.
  \end{align*}
  As $n\to \infty$ we get the desired code family.
  \end{proof}

  \subsubsection{Poltyrev}
  With slight modification to the proof of D'yachkov-Rykov, we can achieve Poltyrev's rate.
  \begin{thm}[Poltyrev \cite{Pol87}] \label{ThmRateBhPol}
  There exist $B_h$-codes of rate $\frac{\log\left(\frac{2^{2h}}{\binom{2h}h}\right)}{2h-1}$.
  \end{thm}
  We need a lemma of basic math.
  \begin{lemma}\label{LemmaBasicMath}
  If $x^{1/n} \le y^{1/m}$ where $0\le x,y\le 1$ and $2\le n\le m$, then $x^{1/(n-1)} \le y^{1/(m-1)}$.
  \end{lemma}
  \begin{proof}
  We have $$x^{1/(n-1)} \le y^{n/(m(n-1))} \le y^{1/(m-1)}.$$
  \end{proof}
  \begin{proof}[Proof of Theorem \ref{ThmRateBhPol}]
  We perform the same random construction as in D'yachkov-Rykov to get $\cC = \{v_1,\ldots,v_t\}\sse \{0,1\}^n$. 
  The multiset $\cC$ may contain several minimal violations.
  For each minimal violation appearing in $\cC$, we arbitrarily pick and remove one vector in this minimal violation.
  In this way we get a set $\cC^\p$ containing no minimal violations.
  
  If $$t = c^\p (\max_{C\in \Conf(\le h,2)} p(C)^{1/(d(C)-1)})^{-n}$$ for some small enough constant $c^\p$, the expected number of minimal violations in $\cC$ is at most $\frac t2$ and the size of $\cC^\p$ is at least $\frac t2$.
  By Lemma \ref{LemmaBasicMath} and Proposition \ref{PropMaxConfh2}, the configuration $C_{\max}(h,2)$ achieves the maximum $p(C)^{1/(d(C)-1)}$ over all $C\in \Conf(\le h,2)$.
  
  So rate of the code $\cC^\p$ is at least
  \begin{align*}
  \frac{\log (t/2)}{n} = (1+o(1)) \frac{-\log p(C_{\max}(h,2))}{d(C_{\max}(h,2))-1} = (1+o(1))\frac{\log\left(\frac{2^{2h}}{\binom{2h}h}\right)}{2h-1}.
  \end{align*}
  As $n\to \infty$ we get the desired code family.
  \end{proof}
  \section{Changing distribution}
  In this section we discuss whether we can change the probability distribution in random constructions of D'yachkov-Rykov and Poltyrev to achieve $B_h$-codes of higher rate.
  
  \subsection{Rate of random coding with a general distribution}
  In the original random construction, $\cC = \{v_1,\ldots,v_t\}$ where each $v_i$ is iid uniformly randomly chosen from $\{0,1\}^n$.
  The strategy we consider is to divide each length-$n$ vector $v_i$ into blocks $v_{i,1}, v_{i,2}, \cdots, v_{i,n/n_0}$ of length $n_0$, where $n_0$ is some constant.
  The $v_{i,j}$'s are iid randomly chosen from a fixed distribution $\cA$ over $\{0,1\}^{n_0}$.
  If $\cA$ is the uniform distribution, then this construction reduces to the original random construction.
  
  \begin{defn}\label{DefnColEnt}
  Let $X$ be a discrete random variable. The collision entropy is defined as $$H_2(X)=-\log \sum_a \bP(X=a)^2.$$
  \end{defn}
  \begin{defn}
  Let $X$ be a random variable. The $n$-fold sum $X^{(h)}$ is a random variable
  such that $X^{(h)} = X_1 + \cdots+X_h$ where $X_i$'s are independent copies of $X$.
  \end{defn}
  \begin{thm}\label{ThmRateBhA}
  Fix a constant $n_0$ and a probability distribution $\cA$ over $\{0,1\}^{n_0}$. Let $X$ be a random variable with distribution $\cA$. 
  Then there exist $B_h$-codes of rate at least
  $\frac{H_2(X^{(h)})}{n_0(2h-1)}$.
  \end{thm}
  Similar to the proof of D'yachkov-Rykov, we need to define configurations to characterize the minimal violations.
  \begin{defn}\label{DefnConfk2A}
  A configuration $C$ of shape $(k,2)$ over distribution $\cA$ is a $k\t 2$ matrix of random variables $(C_{i,j})_{1\le i\le k,1\le j\le 2}$ taking values in $\{0,1\}^{n_0}$ with the following properties.
  \begin{enumerate}
  \item For each $i,j$, $C_{i,j}$ is distributed according to $\cA$.
  \item Some (or no) variables are identified.
  \item No variable appears in two columns.
  \end{enumerate}
  Two configurations of the same shape (and over the same distribution) are equivalent if they have the same law after (1) swapping columns and (2) swapping entries in the same column.
  Let $\Conf_\cA(k,2)$ denote the set of equivalence classes of configurations of shape $(k,2)$ over $\cA$. 
  Let $\Conf_\cA(\le h,2) = \bup_{1\le k\le h} \Conf_\cA(k,2)$.
  \end{defn}
  
  Similar to the uniform distribution case, there is a unique configuration of shape $(h,2)$ over $\cA$ whose all variables are distinct.
  Let $C_{\cA,\max}(h,2)$ denote this configuration.
  
  We prove the following lemmas in analogy with Lemma \ref{LemmaMaxConfh2} and Lemma \ref{LemmaMonoh2}.
  \begin{lemma}\label{LemmaMaxConfh2A}
  Let $X = \sum_{1\le d\le c_i} X_i$ where $c_i\in \bZ_{\ge 1}$, $X_1,\ldots,X_d$ are iid and each $X_i\sim \cA$. Then 
  \begin{align*}
  \sum_a \bP(X=a)^2 \le p(C_{\cA,\max}(d,2)).
  \end{align*}
  \end{lemma}
  \begin{proof}
  Let $Y = \sum_{1\le i\le d} c_i Y_i$ where $Y_1,\ldots,Y_d$ are an independent copy of $X_1,\ldots,X_d$. 
  Then $\sum_a \bP(X=a)^2 = \bP(X=Y)$.
  Considering the characteristic function, we have 
  \begin{align*}
  \bP(X=Y) &= (2\pi)^{-n_0} \int_{[-\pi,\pi]^{n_0}} \phi_{X-Y}(t) d A(t) \\ 
  & = (2\pi)^{-n_0} \int_{[-\pi,\pi]^{n_0}} \prod_{1\le i\le d} \phi_{c_i(X_i-Y_i)}(t) d A(t)\\
  & = (2\pi)^{-n_0} \int_{[-\pi,\pi]^{n_0}} \prod_{1\le i\le d} \phi_{X_i-Y_i}(c_i t) d A(t).
  \end{align*}
  Note that $$\phi_{X_i-Y_i}(t) = \phi_{X_i}(t)\ol {\phi_{Y_i}(t)} = |\phi_{X_1}(t)|^2 \in \bR_{\ge 0}.$$
  So we can apply AM-GM and get
  \begin{align*}
  \bP(X=Y) &\le d^{-1} (2\pi)^{-n_0} \sum_{1\le i\le d} \int_{[-\pi,\pi]^{n_0}} \phi_{X_i-Y_i}(c_it)^d dA(t) \\
  & = d^{-1} (2\pi)^{-n_0} \sum_{1\le i\le d} \int_{[-\pi,\pi]^{n_0}} \phi_{X_i-Y_i}(t)^d dA(t)\\
  & = (2\pi)^{-n_0} \int_{[-\pi,\pi]^{n_0}} \phi_{X_i-Y_i}(t)^d dA(t) \\
  & = (2\pi)^{-n_0} \int_{[-\pi,\pi]^{n_0}} \phi_{X_1+\cdots+X_d-Y_1-\cdots-Y_d}(t) dA(t) \\
  & = \bP(X_1+\cdots+X_d=Y_1+\cdots+Y_d)\\
  & = p(C_{\cA,\max}(d,2)).
  \end{align*}
  \end{proof}
  \begin{lemma}\label{LemmaMonoh2A}
  The value $p(C_{\cA,\max}(d,2))^{1/(2d)}$ is monotone increasing in $d$.
  \end{lemma}
  \begin{proof}
  Let $X_1,\ldots,X_d,Y_1,\ldots,Y_d$ be iid random variables, each with distribution $\cA$.
  Then 
  \begin{align*}
  p(C_{\cA,\max}(d,2)) &= \bP(X_1+\cdots+X_d=Y_1+\cdots+Y_d) \\
  & = (2\pi)^{-n_0} \int_{[-\pi,\pi]^{n_0}} \phi_{X_i-Y_i}(t)^d dA(t).
  \end{align*}
  So the lemma follows from generalized mean inequality.
  \end{proof}
  We have the following proposition in analogy with Proposition \ref{PropMaxConfh2}.
  \begin{prop}\label{PropMaxConfh2A}
  Over all $C\in \Conf_\cA(\le h,2)$, the configuration $C_{\cA,\max}(h,2)$ gives the maximum $p(C)^{1/d(C)}$.
  \end{prop}
  \begin{proof}
  Let $C\in \Conf_\cA(k,2)$ where $1\le k\le h$.
  For $a \in \{0,1\ldots,k\}^{n_0}$, let $p_i(a)$ ($i=1,2$) denote the probability that $C_{1,i} + \cdots + C_{k,i} = a$.
  By Cauchy-Schwarz inequality, we have
  \begin{align*}
  p(C) = \sum_a p_1(a)p_2(a)
  \le \sqrt{(\sum_a p_1(a)^2) (\sum_a p_2(a)^2)}.
  \end{align*}
  
  Let $d_i$ ($i=1,2$) denote the number of distinct variables in column $i$.
  By Lemma \ref{LemmaMaxConfh2A} and Lemma \ref{LemmaMonoh2A},
  \begin{align*}
  \sum_a p_i(a)^2 \le p(C_{\cA,\max}(d_i,2)) \le p(C_{\cA,\max}(h,2))^{d_i/h}.
  \end{align*}
  So we have
  \begin{align*}
  p(C) &\le \sqrt{(\sum_a p_1(a)^2) (\sum_a p_2(a)^2)}\\
  & \le \sqrt{p(C_{\cA,\max}(h,2))^{d_1/h} p(C_{\cA,\max}(h,2))^{d_2/h}} \\
  & = p(C_{\cA,\max}(h,2))^{d(C)/(2h)}.
  \end{align*}
  In other words, 
  \begin{align*}
  p(C)^{1/d(C)} \le p(C_{\cA,\max}(h,2))^{1/(2h)}.
  \end{align*}
  \end{proof}
  
  Now we prove the theorem.
  \begin{proof}[Proof of Theorem \ref{ThmRateBhA}]
  Similar to proof of Theorem \ref{ThmRateBh}, we consider the minimal violations, $i_1,\ldots,i_k,j_1,\ldots,j_k$ ($1\le k\le h$) such that 
  $$v_{i_1} + \cdots+v_{i_k} = v_{j_1} + \cdots + v_{j_k}$$
  and the multisets $\{i_1,\ldots,i_k\}$, $\{j_1,\ldots,j_k\}$ are disjoint.
  For each minimal violation we can associate it with a configuration of shape $(k,2)$ over $\cA$. 
  For each configuration $C\in \Conf_\cA(\le h,2)$, there are $\Th(t^{d(C)})$ minimal violations associated with it, and for each such minimal violation, the probability that it occurs is $p(C)^n$.
  
  Therefore the expected number of minimal violations is at most 
  \begin{align*}
  c \cdot \sum_{C\in \Conf_\cA(\le h,2)} t^{d(C)} p(C)^{n/n_0}
  \end{align*}
  where $c$ is a constant only depending on $h$.
  So when $$t = c^\p (\max_{C\in \Conf_\cA(\le h,2)} p(C)^{1/(d(C)-1)})^{-n/n_0}$$ for some small enough constant $c^\p$, the expected number of minimal violations is at most $\frac t2$.
  So if we remove one vector for each minimal violation, we would get a $B_h$-code $\cC^\p$ of size at least $\frac t2$.
  
  By Proposition \ref{PropMaxConfh2A} and Lemma \ref{LemmaBasicMath}, we have
  $$t = c^\p p(C_{\cA,\max}(h, 2))^{-n/(n_0(d(C_{\cA,\max}(h,2))-1))}$$
  and the rate of code $\cC^\p$ is 
  \begin{align*}
  \frac {\log (t/2)}n = (1+o(1)) \frac{-\log p(C_{\cA,\max}(h,2))}{n_0(d(C_{\cA,\max}(h,2))-1)} = (1+o(1)) \frac{H_2(X^{(h)})}{n_0(2h-1)}
  \end{align*}
  where $X^{(h)} = X_1+\cdots+X_h$ and $X_i$'s are iid random variables with distribution $\cA$.
  
  As $n\to \infty$ we get the desired code family.
  \end{proof}
  \subsection{Maximizing collision entropy}
  In light of Theorem \ref{ThmRateBhA}, if we can find distribution $\cA$ with $$\frac{H_2(X^{(h)})}{n_0} > \log \left(\frac{2^{2h}}{\binom{2h}h}\right),$$
  then we achieve $B_h$-codes with higher rate.
  However, the following result proved by Yury Polyanskiy shows that the uniform distribution is actually optimal.
  Therefore we cannot hope to improve the random coding method by using a clever distribution.
  \begin{prop}[Polyanskiy]\label{PropH2Xh}
  Fix $n_0$. The uniform distribution on $\{0,1\}^{n_0}$ maximizes $\frac{H_2(X^{(h)})}{n_0}$ over all distributions on $\{0,1\}^{n_0}$.
  \end{prop}
  \begin{proof}
  Let $Y^{(h)} = Y_1+\cdots + Y_h$ where $Y_i$'s are independent copies of $X_i$'s.
  Maximizing $H_2(X^{(h)})$ is equivalent to minimizing $\bP(X^{(h)}=Y^{(h)})$.
  Considering the characteristic function, we have
  \begin{align*}
  \bP(X^{(h)}=Y^{(h)}) &= (2\pi)^{-n_0} \int_{[-\pi,\pi]^{n_0}} \phi_{X^{(h)}-Y^{(h)}}(t) dA(t) \\
  &= (2\pi)^{-n_0} \int_{[-\pi,\pi]^{n_0}} \phi_{X}(t)^h\ol{\phi_{X}(t)}^h dA(t) \\
  &= (2\pi)^{-n_0} \int_{[-\pi,\pi]^{n_0}} |\phi_X(t)|^{2h} dA(t).
  \end{align*}
  
  Let $\cS$ be the set of valid probability distributions $\cA$ on $\{0,1\}^{n_0}$.
  Then $\cS$ can be seen as a convex subset of $\bC^{2^{n_0}}$ where each coordinate $a\in \{0,1\}^{n_0}$ represents the probability of $X=a$.
  Let $\cC$ denote the space of continuous functions $f: \bR^{n_0}\to \bC$ with periodicity $2\pi$.
  The map that maps $\cA\in \cS$ to its characteristic function extends to a linear map from $\bC^{2^{n_0}}$ to $\cC$.
  Let $\cT$ be the image of $\cS$ under this map.
  By linearity, $\cT$ is also a convex set.
  
  The $L^{2h}$-norm $$||f||_{2h} = \left((2\pi)^{-n_0}\int_{[-\pi,\pi]^{n_0}} |f(t)|^{2h} dA(t)\right)^{1/(2h)}$$ on $\cC$ is a strictly convex norm.
  Because $\cT$ has at most one intersection point with any line through the origin, the $L^{2h}$-norm is a strictly convex function on $\cT$.
  Therefore its pullback to $\cS$ is also strictly convex.
  Note that this pullback is just $\bP(X^{(h)}=Y^{(h)})^{1/(2h)}$.
  Let us call this function $\rho(\cA)$.
  By symmetry of $\rho$ under $\bF_2^{n_0}$ action, the uniform distribution is a critical point of $\rho$ in $\cS$.
  By strict convexity, the uniform distribution is the unique global minimum of $\rho$.
  \end{proof}
  \section{Random coding for $B_h[g]$-code}
  In this chapter we study the performance of random coding on list-decoding versions of $B_h$-codes. The primary version we consider is the $B_h[g]$-code.
  \subsection{Rate of random coding}
  \begin{defn}
  A $B_h[g]$-code is a set $\cC\sse \{0,1\}^n$ satisfying the property that for any $a\in \{0,1,\ldots,h\}^n$, there exists at most $g$ multisets $\{u_1,\ldots,u_h\}$ such that $a = u_1 + \cdots + u_h$.
  Note that a $B_h[1]$-code is exactly the same as a $B_h$-code.
  The rate of a $B_h[g]$-code is defined as $\frac{\log |\cC|}n$.
  \end{defn}
  
  We would like to apply random coding. Therefore it is important to keep track of the minimal violations. The following definition should not come as a surprise.
  \begin{defn}\label{DefnConfkl}
  A configuration $C$ of shape $(k,l)$ is a $k\t l$ matrix of random variables $(C_{i,j})_{1\le i\le k, 1\le j\le l}$ taking values in $\{0,1\}$ with the property that 
  \begin{enumerate}
  \item For each $i,j$, $\bP(C_{i,j}=0)=\bP(C_{i,j}=1) = \frac 12$.
  \item Some (or no) variables are identified, i.e., $\bP(C_{i,j}=C_{i^\p,j^\p})=1$ for some $i,j,i^\p,j^\p$. We treat identified variables as the same variable. Variables that are not identified are mutually independent.
  \item No variable appears in all columns.
  \item For two different columns, the multiset of variables in this column are different.
  \end{enumerate}
  
  Two configurations of the same shape are equivalent if they have the same law after repeatedly (1) swapping columns and (2) swapping entries in the same column. 
  Let $\Conf(k,l)$ denote the set of equivalence classes of configurations of shape $(k,l)$. 
  Let $\Conf(\le h,l) = \bup_{1\le k\le h} \Conf(k,l)$.
  
  Define $d(C)$ to be the number of distinct variables in $C$. Define $p(C)$ to be the probability that $C_{1,j}+\cdots+C_{k,j}$ are equal for $1\le j\le l$.
  \end{defn}
  \begin{eg}
  There are seven non-equivalent configurations of shape $(2,3)$.
  They are 
  \begin{enumerate}
  \item $C_1=\begin{pmatrix}a & b & c\\a & b & c\end{pmatrix}$. $d(C_1)=3$ and $p(C_1)=\frac 14$.
  \item $C_2=\begin{pmatrix}a & b & c\\a & b & d\end{pmatrix}$. $d(C_2)=4$ and $p(C_2)=\frac 18$.
  \item $C_3=\begin{pmatrix}a & b & d\\a & c & e\end{pmatrix}$. $d(C_3)=5$ and $p(C_3)=\frac 1{16}$.
  \item $C_4=\begin{pmatrix}a & c & e\\b & d & f\end{pmatrix}$. $d(C_4)=6$ and $p(C_4)=\frac{5}{32}$.
  \item $C_5=\begin{pmatrix}a & a & b\\b & c & c\end{pmatrix}$. $d(C_5)=3$ and $p(C_5)=\frac 14$.
  \item $C_6=\begin{pmatrix}a & a & b\\b & c & d\end{pmatrix}$. $d(C_6)=4$ and $p(C_6)=\frac 14$.
  \item $C_7=\begin{pmatrix}a & a & d\\b & c & e\end{pmatrix}$. $d(C_7)=5$ and $p(C_7)=\frac{3}{16}$.
  \end{enumerate}
  Matrix $\begin{pmatrix}a & a & a\\b & c & d\end{pmatrix}$ is not a valid configuration because it violates property 3.
  Matrix $\begin{pmatrix}a & a & c\\b & b & d\end{pmatrix}$ is not a valid configuration because it violates property 4.
  \end{eg}
  
  \begin{thm}\label{ThmRateBhg}
  There exist $B_h[g]$-codes of rate at least
  \begin{align*}
  \min_{C\in \Conf(\le h, g+1)} \frac{-\log p(C)}{d(C)-1}.
  \end{align*}
  \end{thm}
  \begin{proof}
  A violation of the $B_h[g]$-property is a matrix $(x_{i,j})_{1\le i\le h,1\le j\le g+1}$ where $1\le x_{i,j}\le t$, no two columns have the same multisets of variables, and the column sums $$v_{x_{1,j}}+\cdots+v_{x_{h,j}}$$ are equal for all $j$.
  A violation can be non-minimal in the sense that there are numbers appearing in all columns. 
  Therefore the minimal violations we consider are matrices $(x_{i,j})_{1\le i\le k,1\le j\le g+1}$ where $1\le k\le h$, $1\le x_{i,j}\le t$, no two columns have the same multisets of variables and the column sums 
  $$v_{x_{1,j}}+\cdots+v_{x_{k,j}}$$
  are equal for all $j$.
  
  For each minimal violation, we can associate to it a configuration of shape $(k,g+1)$.
  For each configuration $C\in \Conf(\le h,g+1)$, the number of minimal violations associated to it is $\Th(t^{d(C)})$, and each such minimal violation appears with probability $p(C)^n$.
  So the expected number of minimal violations is at most 
  \begin{align*}
  c \cdot \sum_{C\in \Conf(\le h,g+1)} t^{d(C)} p(C)^n
  \end{align*}
  where $c$ is some constant depending only on $h$ and $g$.
  So when 
  \begin{align*}
  t = c^\p \left(\max_{C\in \Conf(\le h,g+1)} p(C)^{1/(d(C)-1)}\right)^{-n}
  \end{align*}
  for some small enough constant $c^\p$, the expected number of violations is no more than $\frac t2$.
  Then we can remove one vector for each minimal violation, and get a $B_h[g]$-code $\cC^\p$ of size at least $\frac t2$.
  
  The rate of code $\cC^\p$ is 
  \begin{align*}
  \frac{\log (t/2)}n = (1+o(1)) \min_{C\in \Conf(\le h,g+1)}\frac{-\log p(C)}{d(C)-1}.
  \end{align*}
  As $n\to \infty$ we get the desired code family.
  \end{proof}
  
  \subsection{Suboptimality of the all-distinct configuration}
  Let $C_{\max}(h,g+1)$ denote the configuration where all variables are distinct.
  In analogy with Proposition \ref{PropMaxConfh2} and Proposition \ref{PropMaxConfh2A}, one may guess that the maximum value of $p(C)^{1/(d(C)-1)}$ is achieved at $C_{\max}(h,g+1)$.
  Unfortunately, this turns out to be not true.
  
  \begin{prop} \label{PropNoMaxConfhg}
  %Fix $g\in \bZ_{\ge 2}$. For $h$ large enough, the configuration $C_{\max}(h,g+1)$ does not give the maximum $p(C)^{1/d(C)}$ over all $C\in \Conf(\le h, g+1)$.
  There exist $g$ and $h$ such that the configuration $C_{\max}(h,g+1)$ does not give the maximum $p(C)^{1/(d(C)-1)}$ over all $C\in \Conf(\le h, g+1)$.
   \end{prop}
%   We need the following lemma.
%   \begin{lemma}\label{LemmaBasicMath2}
%   Fix $h\in \bZ_{\ge 1}$.
%   The value $p(C_{\max}(h,g+1))^{1/d(C_{\max}(h,g+1))}$ is monotone decreasing in $g$.
%   \end{lemma}
%   \begin{proof}
%   \begin{align*}
%   &p(C_{\max}(h,g+1))^{1/d(C_{\max}(h,g+1))} \\
%   & = (2^{-h(g+1)}\sum_{0\le i\le h} \binom hi^{g+1})^{1/(h(g+1))} \\
%   & = \frac 12 (\sum_{0\le i\le h} \binom hi^{g+1})^{1/(h(g+1))}\\
%   & = \frac 12 ||v||_{g+1}^{1/h}
%   \end{align*}
%   where $v$ is the vector $(\binom h0, \binom h1, \cdots, \binom hh)$.
%   The lemma follows from the fact that $||v||_p$ is decreasing in $p$.
%   \end{proof}
  %\begin{proof}[Proof of Proposition \ref{PropNoMaxConfhg}]
  \begin{proof}
  Let $C$ be the following configuration.
  \begin{align*}
  \begin{pmatrix}
  a_1 & c_2 & c_3 & \cdots & c_{g+1}\\
  a_2 & b_2 & b_2 &\cdots &b_2\\
  a_3 & b_3 & b_3 &\cdots &b_3\\
  \vdots & \vdots & \vdots & \vdots & \vdots\\
  a_h & b_h & b_h & \cdots &b_h
  \end{pmatrix}
  \end{align*}
  where different variables denote distinct random variables.
  Clearly $d(C) = 2h-1+g$.
  Column sums are all equal if and only if
  $$a_1+\cdots+a_h = c_2+b_2+\cdots+b_h$$
  and $$c_2=c_3=\cdots=c_{g+1}.$$
  So
  \begin{align*}
  p(C) &= \bP(a_1+\cdots+a_h=c_2+b_2+\cdots+b_h) \bP(c_2=c_3=\cdots=c_{g+1}) \\
  & = 2^{-2h}\binom{2h}h 2^{-(g-1)} \\
  & = \binom{2h}h 2^{-(2h+g-1)}.
  \end{align*}
  So 
  \begin{align*}
  p(C)^{1/(d(C)-1)} &= (2^{-(2h+g-1)} \binom{2h}h)^{1/(2h-1+g-1)}.
%   & = (1+o(1)) \frac 12 \binom{2h}h^{1/(2h)}\\
%   & = (1+o(1)) p(C_{\max}(h,2))^{1/d(C_{\max}(h,2))}
  \end{align*}
  Take $g=2$ and $h=100$. Numerical computations shows that
  \begin{align*}
  p(C)^{1/(d(C)-1)} \approx 0.982312
  \end{align*}
  and 
  \begin{align*}
  p(C_{\max}(h,g+1))^{1/(d(C_{\max}(h,g+1))-1)}\approx 0.981414.
  \end{align*}
  \end{proof}
  \begin{rmk}
  By Lemma \ref{LemmaBasicMath}, the proposition implies that there exist $g$ and $h$ such that the configuration $C_{\max}(h,g+1)$ does not give the maximum $p(C)^{1/d(C)}$ over all $C\in \Conf(\le h, g+1)$.
  
  Numerical computation suggests that for fixed $g\ge 2$, the all-distinct configuration $C_{\max}(h,g+1)$ is suboptimal for all $h$ large enough.
  \end{rmk}
  
  \subsection{Separable configurations}
  Proposition \ref{PropNoMaxConfhg} shows that the rate of random coding construction for $B_h[g]$-codes is much more complicated than that for $B_h$-codes.
  On the other hand, for configurations with nice properties, analogies of Proposition \ref{PropMaxConfh2} and Proposition \ref{PropMaxConfh2A} may hold.
  \begin{defn}
  We say a configuration is separable if no variable appears in two or more columns.
  Let $\SConf(k,l)$ denote the set of separable configurations of shape $(k,l)$.
  Let $\SConf(\le h,l) = \bup_{1\le k\le h} \SConf(k,l)$.
  \end{defn}
  \begin{prop}\label{PropMaxSConfhg}
  Fix $h$ to be an even number.
  Over all separable configurations $C\in \SConf(\le h,g+1)$, the configuration $C_{\max}(h,g+1)$ gives the maximum $p(C)^{1/d(C)}$.
  \end{prop}
  
  We first prove some lemmas. The proofs of them are more difficult than the previous ones. 
  \begin{lemma}\label{LemmaMaxConfhg}
  Let $X = \sum_{1\le i\le d} c_i X_i$ where $c_i\in \bZ_{\ge 1}$, $X_1,\ldots,X_d$ are iid uniform random variables taking values in $\{0,1\}$. Then
  \begin{align*}
  \sum_{a\ge 0} \bP(X=a)^{g+1} \le p(C_{\max}(d,g+1)).
  \end{align*}
  \end{lemma}
  
   Before proving this lemma, let us have some discussions about symmetric decreasing rearrangement and majorization.
  \begin{defn}
  Let $(p_a)_{a\ge 0}$ be a sequence of non-negative numbers with only finitely-many nonzero entries.
  Let $(T(p)_a)_{a\ge 0}$ be the sorted version of $p$ in decreasing order.
  
  For two sequences $p$ and $q$, if we have 
  \begin{align*}
  \sum_{0\le i\le n} T(p)_i \le \sum_{0\le i\le n} T(q)_i,
  \end{align*}
  for all $n$, we say $p$ is majorized by $q$, written as $p\preceq q$.
  
  Let $(S(p)_a)_{a\ge 0}$ be as follows: 
  \begin{align*}
  \cdots \quad T(p)_4 \quad T(p)_2 \quad T(p)_0 \quad T(p)_1 \quad T(p)_3 \quad T(p)_5 \quad \cdots
  \end{align*}
  (with zeros on the left removed).
  We say $S(p)$ is the symmetric decreasing rearrangement of $p$.
  
  For a nonzero integer $c$, define the sequence $C_c(p) = (p_a+p_{a-c})_{a\ge 0}$, where $p_i=0$ for $i<0$.
  \end{defn}
  \begin{lemma}\label{LemmaMaj1}
  Let $(p_a)_{a\ge 0}$ be a sequence of non-negative numbers with only finitely-many nonzero entries. Let $c$ be a nonzero integer.
  Then $C_c(p)\preceq C_1(S(p))$.
  \end{lemma}
  \begin{proof}
  Let $q = T(p)$.
  Simple calculation shows that 
  $$\sum_{0\le i\le n} T(C_1(S(p)))_i =
  \sum_{0\le i\le n-1} q_i
  + \sum_{0\le i\le n+1} q_i.$$
  We know that $\sum_{0\le i\le n} T(C_c(p))_i$ is sum of $2n+2$ terms of $q$, where each $q_i$ appears at most twice.
  So if $$\sum_{0\le i\le n} T(C_1(S(p)))_i < \sum_{0\le i\le n} T(C_c(q))_i,$$
  then $$\sum_{0\le i\le n} T(C_c(p))_i = 2 \sum_{0\le i\le n} q_i$$ and $q_n>q_{n+1}$.
  
  Consider the largest $(n+1)$ terms of $C_c(p)$.
  Each $p_i$ appearing in these terms appears twice.
  Let $a$ be the largest number such that $p_a$ appears in the largest $(n+1)$ terms of $C_c(p)$.
  Then $p_{a-c}$ and $p_{a+c}$ must both appear.
  Because $c\ne 0$, this contradicts the maximality of $a$.
  \end{proof}
  \begin{lemma}\label{LemmaMaj2}
  If $p \preceq q$, then $C_1(S(p))\preceq C_1(S(q))$.
  \end{lemma}
  \begin{proof}
  WLOG assume that $p$ and $q$ are sorted in decreasing order.
  Then for all $n\ge 1$, we have 
  \begin{align*}
  \sum_{0\le i\le n} T(C_1(S(p)))_i &= 
  \sum_{0\le i\le n-1} p_i
  + \sum_{0\le i\le n+1} p_i\\
  &\le \sum_{0\le i\le n-1} q_i
  + \sum_{0\le i\le n+1} q_i\\
  &= \sum_{0\le i\le n} T(C_1(S(q)))_i.
  \end{align*}
  \end{proof}

  \begin{proof}[Proof of Lemma \ref{LemmaMaxConfhg}]
  Let $Y = Y_1+\cdots+Y_d$ where $Y_i$'s are independent copies of $X_i$'s.
  Let $p(Y)$ denote the sequence $(p(Y)_a)_{a\ge 0}$ where $p(Y)_a=\bP(Y=a)$.
  Similarly define $p(X)$.
  The sequences $p(X)$ and $p(Y)$ each contain at most $2^d$ nonzero numbers. 
  The function $f((p_a)_{a\ge 0}) = \sum_a p_a^{g+1}$
  is Schur-convex.
  So we only need to prove that $p(X)\preceq p(Y)$.
  We prove this by induction on $d$.
  
  Base case: When $d=0$, $p(X)=p(Y)$.
  
  Induction step: Suppose the result for $d-1$ variables is true.
  Let $$X^\p = \sum_{1\le i\le d-1} c_i X_i$$ and $$Y^\p = \sum_{1\le i\le d-1} Y_i.$$
  By induction hypothesis, $p(X^\p)\preceq p(Y^\p)$.
  By Lemma \ref{LemmaMaj1} and Lemma \ref{LemmaMaj2}, 
  we have
  \begin{align*}p(X) &= \frac 12 C_c(p(X^\p)) \preceq \frac 12 C_1(S(p(X^\p))) \\
  &\preceq \frac 12 C_1(S(p(Y^\p))) = \frac 12 C_1(p(Y^\p)) = p(Y).
  \end{align*}
  \end{proof}
  
  \begin{lemma}\label{LemmaMonohg}
  Fix $g\in \bZ_{\ge 1}$. Suppose $d\le h$ and $h$ is even. 
  Then we have $$p(C_{\max}(d,g+1))^{1/(d(g+1))} \le p(C_{\max}(h,g+1))^{1/(h(g+1))}.$$
  \end{lemma}
  \begin{proof}
  \begin{align*}
  p(C_{\max}(d,g+1))^{1/(d(g+1))} = \frac 12 (\sum_{0\le i\le d} \binom di^{g+1})^{1/(d(g+1))}.
  \end{align*}
  The sum $\sum_{0\le i\le d} \binom di^{g+1}$ is the constant coefficient of 
  \begin{align*}
  (1+z_1)^d\cdots (1+z_g)^d(1+(z_1\cdots z_g)^{-1})^d.
  \end{align*}
  So by Cauchy's integral formula, we have
  \begin{align*}
  \sum_{0\le i\le d} \binom di^{g+1} &= (2\pi i)^{-g} \oint \cdots \oint (1+z_1)^d \cdots (1+z_g)^d \\&\cdot (1+(z_1\cdots z_g)^{-1})^d (z_1^{-1} dz_1 \cdots z_g^{-1} dz_g)
  \end{align*}
  where the integrals are taken along the unit circles in the complex plane.
  Perform substitution $z_j = \exp(2 i t_j)$. We get 
  \begin{align*}
  \sum_{0\le i\le d} \binom di^{g+1} & = 2^{(g+1)d}\pi^{-g} \\
  &\cdot \int \cdots \int (\cos t_1 \cdots \cos t_g \cos (t_1+\cdots+t_g))^d d t_1 \cdots dt_g
  \end{align*}
  where the integrals are taken over $[-\frac \pi2,\frac\pi2]$.
  Therefore
  \begin{align*}
  &p(C_{\max}(d,g+1))^{1/(d(g+1))} \\& = (\pi^{-g}
  \cdot \int \cdots \int (\cos t_1 \cdots \cos t_g \cos (t_1+\cdots+t_g))^d d t_1 \cdots dt_g)^{1/(d(g+1))} \\
  &\le (\pi^{-g}
  \cdot \int \cdots \int |\cos t_1 \cdots \cos t_g \cos (t_1+\cdots+t_g)|^d d t_1 \cdots dt_g)^{1/(d(g+1))} \\
  &\le (\pi^{-g}
  \cdot \int \cdots \int |\cos t_1 \cdots \cos t_g \cos (t_1+\cdots+t_g)|^h d t_1 \cdots dt_g)^{1/(h(g+1))} \\
  & = (\pi^{-g}
  \cdot \int \cdots \int (\cos t_1 \cdots \cos t_g \cos (t_1+\cdots+t_g))^h d t_1 \cdots dt_g)^{1/(h(g+1))}\\
  & = p(C_{\max}(h,g+1))^{1/(h(g+1))}.
  \end{align*}
  (Third step is generalized mean inequality. Fourth step uses that $h$ is even.) 
  \end{proof}
  \begin{rmk}
  Numerical computation suggests that for fixed $g$, the value $p(C_{\max}(d,g+1))^{1/(d(g+1))}$ is monotone increasing in $d$. If this is indeed true, we can remove the hypothesis that $h$ is even in Proposition \ref{PropMaxSConfhg}.
  \end{rmk}
  
  \begin{proof}[Proof of Proposition \ref{PropMaxSConfhg}]
  Let $C\in \SConf(k,2)$ where $1\le k\le h$. 
  For $a\in \{0,1,\ldots,k\}$, let $p_i(a)$ ($1\le i\le g+1$) denote the probability that $C_{1,i} + \cdots+C_{k,i} = a$.
  By H\"{o}lder's inequality, we have 
  \begin{align*}
  p(C) &= \sum_{0\le a\le k} p_1(a)p_2(a)\cdots p_{g+1}(a)\\
  & \le \prod_{1\le i\le g+1} (\sum_{0\le a\le k} p_i(a)^{g+1})^{1/(g+1)}.
  \end{align*}
  
  Let $d_i$ ($1\le i\le g+1$) denote the number of distinct variables in column $i$.
  By Lemma \ref{LemmaMaxConfhg} and Lemma \ref{LemmaMonohg},
  \begin{align*}
  \sum_{0\le a\le k} p_i(a)^{g+1} \le p(C_{\max}(d_i,g+1)) \le p(C_{\max}(h,g+1))^{d_i/h}.
  \end{align*}
  So we have
  \begin{align*}
  p(C) &\le \prod_{1\le i\le g+1} (\sum_{0\le a\le k} p_i(a)^{g+1})^{1/(g+1)}\\
  &\le \prod_{1\le i\le g+1} p(C_{\max}(h,g+1))^{d_i/(h(g+1))} \\
  &= p(C_{\max}(h,g+1))^{d(C)/(h(g+1))}.
  \end{align*}
  In other words, 
  \begin{align*}
  p(C)^{1/d(C)} \le p(C_{\max}(h,g+1))^{1/(h(g+1))}.
  \end{align*}
  \end{proof}
  \subsection{Another kind of list-decoding}
  There is another natural list-decoding version of $B_h$-codes, which is more closely related to the number of distinct elements in configurations.
  \begin{defn}
  A $B_h^\#[d]$-code is a set $\cC\sse \{0,1\}^n$ satisfying the property that for any $a\in \{0,1,\ldots,h\}^n$, there exists a subset $S\sse \cC$ of size at most $d$ such that
  if $u_1+\cdots+u_h = a$, $u_i\in \cC$ for $1\le i\le h$, then 
  $u_i\in S$ for $1\le i\le h$.
  \end{defn}
  Note that this definition is only meaningful when $d\ge h$.
  Let us apply random coding to this problem. Considering the minimal violations, we have the following definition.
  \begin{defn}
  Define $\Conf^\#(\le h)[d]$ as the set of configurations $C$ of shape $(k,l)$ where $1\le k\le h$, $l\ge 2$, $d(C) \ge d+1-h+k$, and such that removing any column from $C$ will make the number of distinct variables less than or equal to $d-h+k$.
  \end{defn}
  \begin{prop}
  The set $\Conf^\#(\le h)[d]$ is finite.
  \end{prop}
  \begin{proof}
  Let $C\in \Conf^\#(\le h)[d]$. Because every column has at most $h$ distinct variables, the last condition implies that $d(C) \le d+k$. Now we fix $d(C)$.
  For each column, there are a finite number of possible choices. Because $C$ is a valid configuration, no two columns are the same.
  Therefore the number of columns is bounded by a finite number.
  So the number of possible $C$'s is finite.
  \end{proof}
  \begin{thm}
  There exist $B_h^\#[d]$-codes of rate at least
  \begin{align*}
  \min_{C\in \Conf^\#(\le h)[d]} \frac{-\log p(C)}{d(C)-1}.
  \end{align*}
  \end{thm}
  \begin{proof}
  A violation of the $B_h^\#[d]$-property is a matrix $(x_{i,j})_{1\le i\le h, 1\le j\le l}$ with $l\ge 2$ such that the column sums $x_{1,j} + \cdots + x_{h,j}$ are equal for $1\le j\le l$, and the set $\{x_{i,j}, 1\le i\le h, 1\le j\le l\}$ has cardinality at least $d+1$.
  A violation can be non-minimal in the sense that (1) some variables appear in all columns, or (2) we can remove some columns so that the number of distinct $x_{i,j}$'s is still larger than $d$.
  Also, note that removing one occurrence in each column for a variable that appears in all columns will decrease the number of distinct entries by at most one.
  So the restriction on the number of distinct entries is weaker for minimal violations with fewer rows.
  
  Therefore a minimal violation is a matrix $(x_{i,j})_{1\le i\le k, 1\le j\le l}$ with $1\le k\le h$, $l\ge 2$ such that the column sums $x_{1,j} + \cdots + x_{k,j}$ are equal for $1\le j\le l$, the set of $x_{i,j}$'s has cardinality at least $d+1-h+k$, and removing any column will make the matrix have at most $d-h+k$ distinct entries.
  
  For each minimal violation, we can associate to it a configuration in $\Conf^\#(\le h)[d]$. For each such configuration $C$, there are $\Th(t^{d(C)})$ minimal violations associated to it, and each such minimal violation appears with probability $p(C)^n$. So the expected number of minimal violations is at most 
  \begin{align*}
  c \cdot \sum_{C\in \Conf^\#(\le h)[d]} t^{d(C)} p(C)^n
  \end{align*}
  where $c$ is some constant depending only on $h$ and $d$. So whn 
  \begin{align*}
  t = c^\p \left (\max_{C\in \Conf^\#(\le h)[d]} p(C)^{1/(d(C)-1)}\right)^{-n}
  \end{align*}
  for some small enough constant $c^\p$, the expected number of violations is no more than $\frac t2$. Then we can remove one vector for each minimal violation, and get a $B_h^\#[d]$-code $\cC^\p$ of size at least $\frac t2$.
  The rate of code $\cC^\p$ is 
  \begin{align*}
  \frac{\log (t/2)}n = (1+o(1)) \min_{C\in \Conf^\#(\le h)[d]} \frac{-\log p(C)}{d(C)-1}.
  \end{align*}
  As $n\to \infty$ we get the desired code family.
  \end{proof}
  
  \section{Some problems about R\'{e}nyi entropy}\label{ChapRenyi}
  In this chapter we study several problems about maximizing R\'{e}nyi entropy that arises in the study of communication over adder MAC.
  
  We first define R\'{e}nyi entropy, a natural generalization of Shannon entropy.
  \begin{defn}
  Let $X$ be a discrete random variable.
  The R\'{e}nyi entropy of order $\al$ ($\al\ge 0, \al\ne 1$) is defined as 
  \begin{align*}
  H_\al(X) = \frac 1{1-\al} \log \sum_a \bP(X=a)^\al.
  \end{align*}
  The R\'{e}nyi entropy of order $1$ is 
  \begin{align*}
  H_1(X) = \lim_{\al\to 1} H_\al(X) = -\sum_a \bP(X=a)\log \bP(X=a) = H(X)
  \end{align*}
  where $H(X)$ is Shannon entropy.
  The R\'{e}nyi entropy of order $\infty$ (also called min-entropy) is 
  \begin{align*}
  H_\infty(X) = \lim_{\al \to \infty} H_\al(X) = -\log \max_a \bP(X=a).
  \end{align*}
  \end{defn}
  \begin{rmk}
  The R\'{e}nyi entropy of order $2$ is collision entropy (Definition \ref{DefnColEnt}).
  \end{rmk}
  
  In this chapter, the kind of problems we study is the following.
  Fix a set $A$ inside some ambient abelian group and fix a non-negative real number $\al$.
  We would like to determine the maximum $H_\al(X+X)$ where $X$ is a random variable taking value in $A$ (where $X+X$ is understood as sum of two independent copies of $X$).
  We are also interested in the maximum $H_\al(X+Y)$ where $X$ and $Y$ are independent random variables taking value in $A$, and $X$ and $Y$ need not have the same distribution.
  
  \subsection{Addition in $\{0,1\}^n$}\label{SecAdd01n}
  The setting most related to adder MAC is $A = \{0,1\}^n \sse \bZ^n$.
  Ajjanagadde and Polyanskiy \cite{AP15} made the following conjecture arising from studying noisy communication over adder MAC with finite block length.
  \begin{conj}[Ajjanagadde-Polyanskiy \cite{AP15}]
  For $0\le \al \le 1$, the R\'{e}nyi entropy $H_\al(X+Y)$ is maximized at the uniform distribution.
  \end{conj}
  In their conjecture, the distribution of $X$ and $Y$ can be different.
  We consider the same-distribution version and thus it makes sense to generalize Proposition \ref{PropH2Xh} to the following conjecture.
  \begin{conj}\label{ConjHaXX}
  For $0\le \al \le 2$ and $h\ge 2$, the R\'{e}nyi entropy $H_\al(X^{(h)})$ is maximized at the uniform distribution.
  In particular, for $0\le \al \le 2$, the R\'{e}nyi entropy $H_\al(X+X)$ is maximized at the uniform distribution.
  \end{conj}
  \begin{rmk}\label{RmkHaXX}
  The general conjecture is true for $\al=0$ trivially and for $\al=1$ by subadditivity of Shannon entropy. The case $\al=2$ is Proposition \ref{PropH2Xh}.
  \end{rmk}
  We discuss some partial results for the case $h=2$.
  \begin{prop}\label{PropHaXXNonUni}
  For $\al > 2$, there exists $n$ such that the uniform distribution over $\{0,1\}^n$ does not maximize $H_\al(X+X)$.
  \end{prop}
  \begin{proof}
  For $x\in \{0,1\}^n$, denote $\bP(X=x)$ as $p_x$. 
  For $z\in \{0,1,2\}^n$, define $c_z = \sum_{x+y=z} p_x p_y$.
  Define $f(p) = \sum_z c_z^\al$. Then $H_\al(X+X) = \frac 1{1-\al} \log f(p)$.
  Let $p^\circ$ denote the uniform distribution over $\{0,1\}^n$.
  We claim that the uniform distribution does not minimize $f(p)$.
  
  It is easy to see that for $z=\{0,1,2\}^n$, we have
  \begin{align*}
  c_z|_{p^\circ} = 2^{-2n} 2^{\#_1(z)}
  \end{align*}
  where 
  $c_z|_{p^\circ}$ denote $c_z$ for the uniform distribution, 
  and $\#_1(z)$ denote the number of $i$'s ($1\le i\le n$) with $z_i=1$.
  
  Let us compute the first derivatives. For $x\in \{0,1\}^n$ and $z=x+y\in \{0,1,2\}^n$ with $y\in \{0,1\}^n$, we have 
  \begin{align*}
  \frac{\der c_z^\al}{\der p_x} = 2 \al c_z^{\al-1} p_y.
  \end{align*}
  So \begin{align*}
  \frac{\der f(p)}{\der p_x} = \sum_y 2 \al c_{x+y}^{\al-1} p_y.
  \end{align*}
  By symmetry of $f$, the first derivatives $\frac{\der f(p)}{\der p_x}|_{p^\circ}$ are the same for all $x\in \{0,1\}^n$. So we need to check second derivatives.
  
  It is easily computed that 
  \begin{align*}
  \frac{\der^2 c_z^\al}{\der p_x \der p_y} = 4 \al (\al-1) c_z^{\al-2} p_{z-y}p_{z-x}
  \end{align*}
  for $z\ne x+y$ with $z-x,z-y\in \{0,1\}^n$
  and 
  \begin{align*}
  \frac{\der^2 c_z^\al}{\der p_x \der p_y} = 4 \al (\al-1) c_z^{\al-2} p_{z-y}p_{z-x}
  + 2 \al c_z^{\al-1}
  \end{align*}
  for $z=x+y$.
  
  Define $d(x,y)$ to be Hamming distance between $x$ and $y$. 
  For fixed $x$ and $y$, there are $2^{n-d(x,y)}$ different $z$'s such that $z-x,z-y\in \{0,1\}^n$. (For $i$ such that $x_i\ne y_i$, we must have $z_i=1$; for $i$ such that $x_i=y_i$, we have $z_i=i$ or $i+1$.)
  Among these $z$'s, $\binom{n-d(x,y)}{w}$ of them have $\#_1(z) = d(x,y)+w$.
  Therefore
  \begin{align*}
  \frac{\der^2 f(p)}{\der p_x \der p_y}|_{p^\circ} &= 4\al (\al-1) \sum_{0\le w\le n-d(x,y)} \binom {n-d(x,y)} w (2^{d(x,y)+w}2^{-2n})^{\al-2} 2^{-2n} \\
  & + 2 \al (2^{d(x,y)} 2^{-2n})^{\al-1} \\
  & = 4\al (\al-1) (4+2^\al)^{n-d(x,y)} 2^{\al d(x,y) - 2 \al n} + 2 \al (2^{d(x,y)} 2^{-2n})^{\al-1}.
  \end{align*}
  
  Let $A$ be an $2^n\t 2^n$ matrix indexed by $\{0,1\}^n$ with $$A_{x,y} = \frac{\der^2 f(p)}{\der p_x \der p_y}|_{p^\circ}.$$
  We claim that there exists a length-$2^n$ vector $v$ with $\sum_x v_x=0$ and $v^t A v < 0$. 
  Let $v$ be such that $v_x = (-1)^{\#_1(x)}$. Then we have
  \begin{align*}
  v^t A v &= \sum_{x,y\in \{0,1\}^n} (-1)^{d(x,y)} A_{x,y} \\
  & = 2^n \sum_{0\le d\le n} \binom nd (-1)^d (4\al (\al-1) (4+2^\al)^{n-d} 2^{\al d - 2 \al n} + 2 \al (2^d 2^{-2n})^{\al-1})\\
  & = 2^{1+2n-2\al n} ((2-2^\al)^n + 2^{n+1} (\al-1))\al.
  \end{align*}
  If $\al > 2$ and $n$ is a large enough odd number, the above value is negative.
  So $f(p^\circ + \ep v) < f(p^\circ)$ for $\ep>0$ small enough.
  \end{proof}
%   \begin{lemma}\label{LemmaPSD}
%   Let $A$ be a symmetric real matrix with eigenvector $\bbl$ with eigenvalue $\lm$.
%   Then $A$ is positive semidefinite (resp.~positive definite) if and only if $\lm\ge 0$ (resp.~$\lm>0$) and $v^t A v \ge 0$ (resp.~$v^t A v > 0$) for all $v$ with $v^t \bbl = 0$.
%   \end{lemma}
%   \begin{proof}
%   The ``only if'' part is obvious. Let us prove the ``if'' part.
%   Any vector $u$ can be written as $v + c \bbl$ for some $c\in \bR$. So
%   \begin{align*}
%   u^t A u &= v^t A v + 2c v^t A \bbl + c^2 \bbl^t A \bbl \\
%   & = v^t A v + 2c\lm v^t \bbl + c^2 \lm ||\bbl||_2^2\\
%   & = v^t A v + c^2 \lm ||\bbl||_2^2.
%   \end{align*}
%   So if $v^t A v \ge 0$ (resp.~$v^t A v > 0$), we have $u^t A u\ge 0$ (resp.~$u^t A u> 0$).
%   \end{proof}
  \begin{prop}\label{PropHaXXUni}
  For $0\le \al \le 2$, the uniform distribution over $\{0,1\}^n$ is a local maximum of $H_\al(X+X)$.
  \end{prop}
  \begin{proof}
  The cases $\al=0$ and $\al=1$ follow from Remark \ref{RmkHaXX}.
  Now assume $\al\ne 0,1$. 
  Follow notations $p_x, c_z, p^\circ, f, A$ in proof of Proposition \ref{PropHaXXNonUni}.
  
  We prove that $p^\circ$ is a local maximum of $f(p)$ for $0<\al<1$ and a local minimum of $f(p)$ for $1<\al<2$.
  By proof of Proposition \ref{PropHaXXNonUni}, the first derivatives 
  $\frac{\der f(p)}{\der p_x}|_{p^\circ}$ are the same for all $x\in \{0,1\}^n$, 
  and 
  \begin{align*}
  \frac{\der^2 f(p)}{\der p_x \der p_y}|_{p^\circ} = 
  4\al (\al-1) (4+2^\al)^{n-d(x,y)} 2^{\al d(x,y) - 2 \al n} + 2 \al (2^{d(x,y)} 2^{-2n})^{\al-1}.
  \end{align*}
  Note that $\bbl$ is an eigenvector of $A$.
  Let $\bbl^\perp$ denote the vector space of vectors $v$ orthogonal to $\bbl$, i.e., $v^t \bbl =0$. Clearly $A$ acts on $\bbl^\perp$.
  
  We prove that matrix $A$ is positive definite (resp.~negative definite) on $\bbl^\perp$ for $1<\al<2$ (resp.~$0<\al<1$).
  Let $v\in \bbl^\perp$ be an eigenvector of $A$ with eigenvalue $\lm$.
  Note that $A_{x,y}$ only depends on $d(x,y)$, so $A$ possesses a lot of symmetry.
  Let $g_i : \{0,1\}^n \to \{0,1\}^n$ ($1\le i\le n$) be the map that flips the $i$-th coordinate.
  Then $g_i^{-1} A g_i = A$.
  Therefore $g_i v$ is also an eigenvector of $A$ with eigenvalue $\lm$.
  
  We repeatedly perform the following: Choose a coordinate $i$ such that $v \ne g_i v$ and $v+g_i v\ne 0$, and replace $v$ with $v + g_i v$.
  This process ends in at most $n$ turns, and when it ends, the vector $v$ satisfies 
  the property that for each coordinate $i$, either $v = g_i v$ or $v = -g_i v$.
  
  Suppose there are $m$-coordinates $i$ with $v = -g_i v$. 
  Because $v\in \bbl^\perp$, $m\ne 0$.
  WLOG assume that for $i=1,\ldots,m$, $v=-g_i v$.
  By multiplying by a nonzero constant, we can assume that
  $$v_x = (-1)^{x_1 + \cdots + x_m}.$$
  Then we can compute
  \begin{align*}
  v^t A v &= \sum_{x,y\in \{0,1\}^n} (-1)^{d(x,y)} A_{x,y} \\
  & = 2^{2n-m} \sum_{0\le d\le m} \binom md (-1)^d  (4\al (\al-1) (4+2^\al)^{n-d} 2^{\al d - 2 \al n} + 2 \al (2^d 2^{-2n})^{\al-1})\\
  & = 2^{1 - m + 4n - 2 \al n} \al ((1 - 2^{\al-1})^m + (1 + 2^{\al-2})^{n-m} (2\al-2)).
  \end{align*}
  So it remains to study the function $$g_\al(n,m) = (1 - 2^{\al-1})^m + (1 + 2^{\al-2})^{n-m} (2\al-2).$$
  When $1<\al<2$, we have 
  \begin{align*}
  g_\al(n,m) &\ge g_\al(m,m) = (1-2^{\al-1})^m + 2\al-2 \\
  & \ge 2\al-2 - |1-2^{\al-1}| = 2 \al - 1 - 2^{\al-1} =: h(\al).
  \end{align*}
  When $0<\al<1$, we have
  \begin{align*}
  g_\al(n,m) &\le g_\al(m,m) = (1-2^{\al-1})^m + 2\al-2 \\
  & \le 2\al-2 + |1-2^{\al-1}| = 2\al - 1 - 2^{\al-1} = h(\al).
  \end{align*}
  It remains to show that $h(\al)<0$ for $0<\al<1$ and $h(\al)>0$ for $1<\al<2$.
  This follows from $h(1)=0$ and 
  \begin{align*}
  h^\p(\al) = 2 - 2^{\al-1} \log_e 2 > 0
  \end{align*}
  on the interval $[0,2]$.
  \end{proof}
  \subsection{Addition in a Sidon set}
  In Section \ref{SecAdd01n}, the additive structure of $\{0,1\}^n$ can be thought of as a source of complexity of the problem. Therefore it is natural to consider addition over a set with minimal additive structure, such as Sidon sets ($B_2$-sets in Definition \ref{DefnBhSet}) in some ambient abelian group.
  Ganesh Ajjanagadde, in private communication, made the following conjecture.
  \begin{conj}
  If $A$ is a Sidon set, then the $H(X+Y)$ achieves its maximum at uniform distribution.
  \end{conj}
  We consider the same-distribution version with R\'{e}nyi entropy, and prove the following results.
  \begin{prop}
  Let $\al_*$ be the unique root of the equation $$2^\al \al -4\al +2 =0$$ in range $[1.1,2]$ (with approximate value $\al_* \approx 1.29856$).
  For $0\le \al \le \al_*$, if $A$ is a Sidon set, then the R\'{e}nyi entropy $H_\al(X+X)$ achieves its maximum at uniform distribution.
  \end{prop}
  \begin{proof}
  The case $\al=0$ is obvious.
  We first prove the case $\al=1$.
  For $x\in A$, we denote $P(X=x)$ as $p_x$.
  We have 
  \begin{align*}
  -H(X+X) = \sum_x p_x^2 \log (p_x^2) + \sum_{x < y} 2 p_x p_y \log (2 p_x p_y)
  \end{align*}
  where $<$ is an arbitrary total order on $\{0,1\}^n$.
  Let $f(p) = -H(X+X)$. Our goal is to minimize $f(p)$.
  Let us compute the first derivative.
  \begin{align*}
  \frac {\der f(p)}{\der p_x} &= 2p_x \log e + 2p_x \log (p_x^2) + \sum_{y\ne x} (2p_y \log e + 2p_y\log(2p_x p_y)) \\
  & = 2 \log e - 2 p_x + \sum_{y \in A} 2p_y \log(2p_x p_y) \\
  & = 2 \log e - 2 p_x + 2  + 2 \log p_x + 2 \sum_{y \in A} p_y \log(p_y) \\
  & = 2 \log e - 2 p_x + 2  + 2 \log p_x - 2 H(X).
  \end{align*}
  The function $-2 p_x + 2\log p_x$ is monotone increasing in $p_x\in [0,1]$.
  Therefore if there exists $x,y\in A$ such that $p_x < p_y$, then we can make the transform $p_x \mto p_x+\ep$, $p_y\mto p_y-\ep$ for some small $\ep>0$ so that $f(p)$ decreases.
  So a local minimum point of $f(p)$ must be the uniform distribution.
  
  Now we consider the case $\al\ne 1$.
  Let $$f(p) = \sum_x p_x^{2\al} + \sum_{x<y} (2p_xp_y)^{\al}.$$
  Then $H_\al(p) = \frac 1{1-\al} f(p)$.
  We would like to maximize $f(p)$ when $0<\al<1$ and minimize $f(p)$ when $1<\al<\al_*$.
  Let us compute the first derivative.
  \begin{align*}
  \frac{\der f(p)}{\der p_x} &= 2\al p_x^{2\al-1} + \sum_{y\ne x} \al 2^\al p_y^\al p_x^{\al-1} \\
  & = \al p_x^{\al-1} (2 p_x^\al + \sum_{y\ne x} 2^\al p_y^\al) \\
  & = \al p_x^{\al-1} ((2-2^\al) p_x^\al + \sum_{y\in A} 2^\al p_y^\al).
  \end{align*}
  Let $B = \sum_{y\in A} 2^\al p_y^\al$. Then clearly $B \ge 2^\al p_x^\al$.
  If we view $B$ as a constant, then
  \begin{align*}
  &\frac{\der}{\der p_x} (\al p_x^{\al-1} ((2-2^\al) p_x^\al + B)) \\
  & = \al p_x^{\al-2} (B(\al-1) - (2^\al-2) (2\al-1) p_x^\al).
  \end{align*}
  Using $B \ge 2^\al p_x^\al$, we see that
  $$\al p_x^{\al-1} ((2-2^\al) p_x^\al + B)$$ is monotone increasing in $p_x$ when $1<\al<\al_*$, and is monotone decreasing in $p_x$ when $0<\al<1$.
  Therefore if there are $p_x<p_y$, we can make transform $p_x\mto p_x+\ep$, $p_y\mto p_y-\ep$ for some small $\ep>0$  so that $f(p)$ decreases (when $1<\al<\al_*$) or increases (when $0<\al<1$).
  So a local maximum point of $H_\al(X+X)$ must be the uniform distribution.
  \end{proof}
  \begin{prop}
  Let $\al^*$ be the unique root of the equation $$2^\al-4\al+2=0$$
  in range $[3, 4]$ (with approximate value $\al^*\approx 3.65986$).
  For $\al>\al^*$, for some Sidon set $A$, the R\'{e}nyi entropy $H_\al(X+X)$ does not achieve its maximum at uniform distribution.
  \end{prop}
  \begin{proof}
  Let $A = \{0,1\}$ be a Sidon set with two elements.
  Let $p=\bP(X=1)$. Then $1-p=\bP(X=0)$.
  Let $$f(p) = p^{2\al} + (2p(1-p))^\al + (1-p)^{2\al}.$$
  Then $H_\al(X+X) = \frac 1{1-\al} \log f(p).$
  For $\al>\al^*$, maximizing $H_\al(X+X)$ is equivalent to minimizing $f(p)$.
  
  Simple calculation shows that 
  $f^\p(\frac 12)=0$ and $$f^\pp(\frac 12) = -2^{3-2\al} (2^\al -4\al +2)\al.$$
  When $\al > \al^*$, we have $f^\pp(\frac 12)<0$, and thus $p=\frac 12$ is not a local minimum of $f$.
  \end{proof}
  
  \bibliographystyle{alpha}
  \bibliography{main}
\end{document}